\documentclass[a4paper,11pt]{article}

\usepackage{a4wide}
\usepackage{fullpage}
\usepackage{amsmath}
\usepackage{amsfonts}
\usepackage{amssymb}
\usepackage{amsthm}
\usepackage{thm-restate}

\usepackage{tabularx,array,booktabs,multirow}
\usepackage{enumitem}

\usepackage{url}

\usepackage{graphicx}
\usepackage{caption}
\usepackage{subcaption}
\usepackage{tikz}
\usetikzlibrary{arrows,decorations.pathmorphing,decorations.pathreplacing,backgrounds,positioning,fit,matrix}
\usetikzlibrary{shapes,calc}

\usepackage{xcolor}

\usepackage[ruled]{algorithm}
\usepackage[noend]{algpseudocode}
\algnewcommand\algorithmicinput{\textbf{Input:}}
\algnewcommand\algorithmicoutput{\textbf{Task:}}
\algnewcommand\Input{\item[\algorithmicinput]}
\algnewcommand\Output{\item[\algorithmicoutput]}

\usepackage{todonotes}
\usepackage{xspace}

\usepackage[pdfdisplaydoctitle,colorlinks]{hyperref}
\usepackage[capitalize, nameinlink]{cleveref}

\newcommand{\NN}{\mathbb{N}} 
 
\newcommand{\sDCParas}{\ell}
\newcommand{\PVCParak}{k}

\newcommand{\PVCParas}{\sDCParas}
\newcommand{\PVCParal}{\PVCParas}

\newcommand\abs[1]{|#1|} 

\newcommand{\problemdef}[4]{
\begin{center}
  \begin{minipage}{0.95\textwidth}
    \normalsize\textsc{#2} \smallskip \\
    \begin{tabularx}{\textwidth}{@{}l@{\hspace{3pt}}X}
      \normalsize\textbf{Input:} & \normalsize#3 \\
      \normalsize\textbf{#1:}    & \normalsize#4
    \end{tabularx}
  \end{minipage}
\end{center}
}

\newtheorem{theorem}{Theorem}
\newtheorem{lemma}{Lemma}

\newtheorem{rrule}{Reduction Rule}
\crefname{obs}{observation}{Observation}
\crefname{rrule}{Reduction Rule}{Reduction Rules}

\begin{document}

\setlength{\abovedisplayskip}{0pt}
\setlength{\abovedisplayshortskip}{0pt}

\title{Kernelization for Partial Vertex Cover via (Additive) Expansion Lemma}

\author{Tomohiro Koana \and André Nichterlein \and Niklas Wünsche}

\date{TU Berlin, Faculty IV, Algorithmics and Computational Complexity, Berlin, Germany\\
\texttt{\{tomohiro.koana,andre.nichterlein\}@tu-berlin.de}}

\maketitle

\begin{abstract}
	Given a graph and two integers~$k$ and~$\ell$, \textsc{Partial Vertex Cover} asks for a set of at most~$k$ vertices whose deletion results in a graph with at most~$\ell$ edges.
	Based on the expansion lemma, we provide a problem kernel with~$(\ell + 2)(k + \ell)$ vertices.
	We then introduce a new, additive version of the expansion lemma and show it can be used to prove a kernel with~$(\ell + 1)(k + \ell)$ vertices for~$\ell \ge 1$.
\end{abstract}

\section{Introduction}

\textsc{Vertex Cover} is arguably the most studied problem in parameterized algorithmics, serving as example for many techniques~\cite{CFKL+15}. 
Given a graph and an integer~$k$, the question is whether there are~$k$ vertices covering all edges, i.\,e., the removal of said~$k$ vertices results in an edgeless graph.
One of the most important techniques from parameterized algorithmics is kernelization.
In a nutshell, the question is whether there is efficient and effective data reduction for a given problem.
The best known kernelization algorithm \textsc{Vertex Cover} takes as input an instance and a constant~$c$ and computes in polynomial time an equivalent instance with at most~$2k - c \log k$ vertices~\cite{Lam11}.

Also studying kernelization, we consider \textsc{Partial Vertex Cover}:
Given a graph and two integers~$k$ and~$\ell$, the question is whether there are at most~$k$ vertices whose deletion results in a graph with at most~$\ell$ edges.
\textsc{Partial Vertex Cover} is well-studied in the parameterized literature:
It is W[1]-hard with respect to~$k$~\cite{GNW07}, but fixed-parameter tractable with respect to~$k+\ell$~\cite{RS08}, $k$ plus the degeneracy of the input graph~\cite{AFS11,PY22,KKNS22}, the number of covered edges~$m-\ell$~\cite{KLR08}, and~$k$ plus the $c$-closure of the graph~\cite{KKNS22}.
\textsc{Partial Vertex Cover} is also studied on special graph classes: Unlike \textsc{Vertex Cover}, it is NP-hard~\cite{CMPS14} but fixed-parameter tractable with respect to~$k$ on bipartite graphs~\cite{AFS11,MPSW20}.
There are subexponential-time algorithms with respect to~$k$ on planar and apex-minor graphs~\cite{FRS11}.
In graphs of constant degeneracy or $c$-closure, \textsc{Partial Vertex Cover} admits kernels of size~$k^{O(1)}$~\cite{KKNS22}.

To the best of our knowledge, the kernelization of \textsc{Partial Vertex Cover} with respect to~$k + \ell$ has not been studied.
We fill this gap by providing polynomial kernels.
More precisely, using the expansion lemma we provide a kernel with~$(\ell + 2)(k + \ell)$ vertices.
In a second step, we present an \emph{additive} version of the expansion lemma and show how it can be used to obtain a smaller kernel with~$(\ell + 1)(k + \ell)$ vertices for~$\ell \ge 1$.

\section{Preliminaries}

\paragraph{Kernelization}
A \emph{parameterized problem} is a set of instances~$(I,k)$ where~$I \in\Sigma^*$ for a finite alphabet~$\Sigma$ and~$k\in \mathbb{N}$ is the \emph{parameter}.
We say that two instances~$(I,k)$ and $(I',k')$ of a parameterized problem~$P$ are \emph{equivalent} if~$(I,k) \in P$ if and only if~$(I',k') \in P$. 
A \emph{kernelization} is an algorithm that, given an instance~$(I,k)$ of a parameterized problem~$P$, computes in polynomial time an equivalent instance~$(I',k')$ of~$P$ (the \emph{kernel}) such that $|I'|+k'\leq f(k)$ for some  computable function~$f$. %
We say that~$f$ measures the \emph{size} of the kernel, and if~$f(k)\in k^{O(1)}$, then we say that $P$~admits a polynomial kernel. 
Typically, a kernel is achieved by applying polynomial-time executable data reduction rules.
We call a data reduction rule~$\mathcal{R}$ \emph{safe} if the new instance~$(I',k')$ that results from applying~$\mathcal{R}$ to~$(I,k)$ is equivalent to~$(I,k)$.

\paragraph{Notation}
We use standard notation from graph theory. 
All graphs considered in this work are simple and undirected.
For a graph~$G$, we denote with~$V(G)$ the vertex set and~$E(G)$ the edge set with the respective sizes~$n := |V(G)|$ and~$m := |E(G)|$.
For~$v \in V(G)$, we denote with~$N(v)$ the set of its neighbors.
Let~$U \subseteq V(G)$.
We denote with~$G[U]$ the subgraph induced by~$U$ and set~$G-U := G[V(G) \setminus U]$.

\problemdef{Question}{Partial Vertex Cover}
{An undirected graph~$G$ and~$k,\ell \in N$}
{Is there a size-$k$ vertex set~$S$ such that~$G-S$ has at most~$\ell$ edges?}

\section{Kernelization with respect to~$k+\ell$}

Chen et al.~\cite{CKJ01} observed that the theorem of Nemhauser and Trotter \cite{NT75} yields a~$2k$-vertex kernel for \textsc{Vertex Cover}.
This is based on the LP relaxation of \textsc{Vertex Cover} (we will refer to it as VCLP for short), which can be formulated as an integer linear program as follows using a variable~$x_v$ for each~$v \in V(G)$:

\begin{align*}
    \min \sum_{v \in V(G)} x_v \qquad\text{subject to}\quad & x_u + x_v \ge 1 \quad \forall uv \in E(G),  \\[-3ex]
                                                       & x_v \in \{ 0, 1 \} \quad \forall v \in V(G).
\end{align*}
In VCLP, the last integral constraint is replaced with~$0 \le x_v \le 1$ for each~$v \in V(G)$.
It is well-known that VCLP admits an optimal solution~$(x_v)_{v \in V(G)}$ such that~$x_v \in \{ 0, 1/2, 1 \}$ for each~$v \in V(G)$.
Moreover, such a solution can be computed in~$O(m \sqrt{n})$ time by computing a maximum matching in an auxiliary bipartite graph (see, for instance, \cite[Section 2.5]{CFKL+15}).
For a half-integral optimal solution~$(x_v)_{v \in V(G)}$, let~$V_0 := \{ v \in V(G) \mid x_v = 0 \}$, $V_1 := \{ v \in V(G) \mid x_v = 1 \}$, and~$V_{1/2} := \{ v \in V(G) \mid x_v = 1 / 2 \}$.
Nemhauser-Trotter theorem \cite{NT75} states that there exists a minimum vertex cover~$S$ such that~$V_0 \cap S = \emptyset$ and~$V_1 \subseteq S$, which leads to the following reduction rule:
Delete~$V_0 \cup V_1$ and decrease~$k$ by~$|V_1|$.
A simple analysis shows that this results in a kernel with at most $2k$~vertices for \textsc{Vertex Cover}.
When it comes to \textsc{Partial Vertex Cover}, however, this reduction rule is no longer safe.
In fact, it is not difficult to construct instances where~$S \cap V_1 = \emptyset$ and~$V_0 \subseteq S$ for every solution~$S$, see \cref{fig:counter-example}.
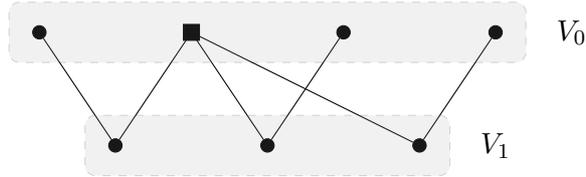
\begin{figure}
	\centering
	\tikzstyle{knoten}=[circle,minimum size=5pt,draw,inner sep=1pt,fill=black]
	\tikzstyle{majarr}=[draw=black]
	\begin{tikzpicture}[auto]
		\foreach \i in {1,2,3} {
			\node[knoten] at (2*\i ,0) (i\i) {};
		}
		\foreach \i in {1,2,3,4} {
			\node[knoten] at (2*\i - 1, 1.5) (h\i) {};
		}
	
		\foreach \i / \j in {1/1, 1/2, 2/2, 2/3, 3/2, 3/4} {
			\draw[majarr] (i\i) edge (h\j);
		}

		\node at (8,1.5) {$V_0$};
		\node at (7,0) {$V_1$};
		\node[minimum size=6pt,draw,inner sep=1pt,fill=black] at (3,1.5) {};
		
		\def\distance{0.4}
		
		\foreach \x / \y / \xx in {2/0/6, 1/1.5/7} {
			\path[fill=black!35,draw=black,dashed,opacity=.15,rounded corners] (\x - \distance, \y - \distance) rectangle ( \xx + \distance, \y + \distance);
		}
	\end{tikzpicture}
	\caption{A graph with seven vertices. The top four vertices are in~$V_0$, the bottom three are in~$V_1$. For~$k=1$ and~$\ell = 3$ there is a unique solution: take the square-marked vertex in~$V_0$ and leave the three edges incident to the three other vertices in~$V_0$ uncovered.}
	\label{fig:counter-example}
\end{figure}
Instead, we use the expansion lemma to identify some vertices we may safely delete from~$V_0 \cup V_1$.

We will assume throughout that every vertex is non-isolated because an isolated vertex can be simply deleted:

\begin{rrule}
  \label{rr:deg0}
  If a vertex~$v \in V(G)$ is isolated, then delete~$v$.
\end{rrule}

Our next reduction rule bounds the size of~$V_1$ and~$V_{1/2}$ in terms of~$k$ and~$\ell$.

\begin{rrule} \label{rrule:LinKernelRRuleV1V12}
If~$\abs{V_1}  + 1/2 \cdot \abs{V_{1/2}} > \PVCParak + \PVCParas$, then return No. 
\end{rrule}

\begin{lemma}\label{lemma:VCPVCSizeBound}
  \Cref{rrule:LinKernelRRuleV1V12} is safe.
\end{lemma}
\begin{proof}
Suppose that~$|V_1| + 1/2 \cdot \abs{V_{1/2}} > \PVCParak + \PVCParas$.
Assume for contradiction that there is a set~$S \subseteq V(G)$ such that~$|S| \le k$ and~$|E(G - S)| \le \ell$.
Then, $G$ has a vertex cover of size at most~$k + \ell$:
we choose an arbitrary endpoint of each edge~$e \in E(G - S)$ and add it to~$S$.
The resulting set~$S'$ is a vertex cover of~$G$ of size at most~$\PVCParak + \PVCParal$.
However, since the optimum of VCLP is~$|V_1| + 1/2 \cdot |V_{1/2}| > k + \ell$, every vertex cover should have size greater than~$k + \ell$.
Thus, we have a contradiction.
\end{proof}

In the following, we show two ways to bound the size~$V_0$.
In \Cref{sec:expansion}, we show that~$|V_0|$ can be bounded by~$(\ell + 1) |V_1|$ using the expansion lemma \cite{CFKL+15}.
In \Cref{sec:add-expansion}, we improve the upper bound to~$\ell |V_1|$ using what we call the \emph{additive expansion lemma}. 

\subsection{Kernel using expansion lemma}
\label{sec:expansion}

In this subsection, we bound the size of~$V_0$ using the expansion lemma.
Let~$H$ be a bipartite graph with bipartition~$(A, B)$ and let~$q \in \NN_+$.
We say that a set of edges~$M$ is a~$q$-expansion from~$A$ into~$B$ if every vertex of~$A$ is incident to exactly~$q$ edges of~$M$ and every vertex of~$B$ is incident to at most one edge of~$M$.
It is known that there is a~$q$-expansion from~$A$ into~$B$ if and only if~$|N(X)| \ge q \cdot |X|$ for every~$X \subseteq A$ \cite{CFKL+15}.

%

\begin{lemma}[Expansion lemma {\cite[Lemma 2.18]{CFKL+15}}]\label{lemma:expansion} 
	Let~$q \in \NN_+$ and~$H$ be a bipartite graph with bipartition~$(A,B)$. If
	\begin{itemize}
		\item $\abs{B} \geq q \cdot \abs{A}$, and
		\item there are no isolated vertices in $B$, 
	\end{itemize}
	then there exist non-empty vertex sets $X \subseteq A$ and $Y \subseteq B$ such that
	\begin{itemize}
		\item there is a $q$-expansion of $X$ into $Y$, and
		\item no vertex in $Y$ has a neighbor outside $X$, that is,  $N(Y) \subseteq X$.
	\end{itemize}
	Furthermore, the sets $X$ and $Y$ are computable in polynomial time. 
\end{lemma}

Generalizing the data reduction rule for \textsc{Vertex Cover} based on crowns (a crown is a~$1$-expansion), we get the following data reduction rule.

\begin{rrule}\label{rrule:LinKernelRRuleV0} 
	If there is an $(\ell + 1)$-expansion~$M$ from~$X$ into~$Y$ in~$G[V_1, V_0]$ such that~$N(Y) \subseteq X$, then delete~$X$ and~$Y$ from~$G$ and decrease~$k$ by~$|X|$.
\end{rrule}

\begin{lemma}
  \Cref{rrule:LinKernelRRuleV0} is safe. 
\end{lemma}
\begin{proof}
	Let~$(G', k', \ell)$ denote the resulting instance.
	
	($\Rightarrow$): 
	Let~$S$ be a solution of~$(G, k, \ell)$.
	We claim that~$S' = S \setminus (X \cup Y)$ is a solution of~$(G', k', \ell)$.
	First, observe that~$|E(G' - S')| \le |E(G - S)| \le \ell$ since~$G'-S'$ is a subgraph of~$G-S$.
	We then show that~$|S'| \le k' = k - |X|$.
	Observe that for every~$x \in X \setminus S$, $S$ contains at least one vertex of~$M(x) \cap Y$, where~$M(x)$ is the set of vertices matched to~$x$ in~$M$, since~$|M(x)| = \ell + 1$.
	Thus, we have~$|S'| \le |S| - |X| \le k'$, showing that~$S'$ is a solution of~$(G', k', \ell)$.

	($\Leftarrow$)
	Let~$S'$ be a solution of~$(G', k', \ell)$.
	Because~$N(y) \subseteq X$ for every~$y \in Y$, we have~$|E(G - (S' \cup X))| = |E(G' - S')| \le \ell$.
	Thus, $S' \cup X$ is a solution of~$(G, k, \ell)$ as~$|S' \cup X| \le k$.
\end{proof}

\begin{theorem}
  \textsc{Partial Vertex Cover} admits kernel with at most~$(\ell + 2)(k + \ell)$ vertices.
\end{theorem}
\begin{proof}
	Given an instance~$(G, k, \ell)$ of \textsc{Partial Vertex Cover}, we first compute a half integral solution of VCLP and apply \Cref{rrule:LinKernelRRuleV1V12}.
	As long as~$|V_0| \ge (\ell + 1) |V_1|$, we apply \Cref{lemma:expansion} on~$G[V_1, V_0]$ with~$q = \ell + 1$, get an~$(\ell + 1)$-expansion from~$X$ into~$Y$, and apply \Cref{rrule:LinKernelRRuleV0}.
	
	We then end up with an instance with $|V_0| \le (\ell + 1) |V_1|$.
	Thus, the number of remaining vertices is at most

	\begin{align*}
		|V_1| + |V_{1/2}| + |V_0| < (\ell + 2)|V_1| + |V_{1/2}| \le \ell \cdot |V_1| + 2 (|V_1| + 1/2 \cdot |V_{1/2}|).
	\end{align*}
	Since~$|V_1| \le |V_1| + 1/2 \cdot |V_{1/2}| \le k + \ell$ after applying \Cref{rrule:LinKernelRRuleV1V12}, we have at most~$(\ell + 2) (k + \ell)$ vertices.
\end{proof}

\subsection{Smaller kernel using additive expansion lemma}
\label{sec:add-expansion}

In this subsection, we show that the kernel size can be improved using what we call \emph{additive expansions} instead.

Let~$H$ be a bipartite graph with bipartition~$(A, B)$ and let~$q \in \NN_+$.
We say that~$H$ is a \emph{$q$-additive expansion} from~$A$ into~$B$ if for every~$B' \subseteq B$ of size~$q$, there is a matching saturating~$A$ in~$H[A, B \setminus B']$.
It follows from Hall's theorem that there is a~$q$-additive expansion if and only if~$|N(X)| \ge |X| + q$ for every nonempty~$X \subseteq A$.
We remark that every~$(q + 1)$-expansion is a~$q$-additive expansion.

We show a lemma analogous to the expansion lemma.
Our proof relies on the polynomial-time solvability of submodular minimization. 

\begin{lemma}
	\label{lemma:additive-expansion}
	Let~$H$ be a bipartite graph with a bipartition~$(A, B)$ and~$q \in \NN_+$.
	There is a polynomial-time algorithm that finds a nonempty set~$X \subseteq A$ with~$|N(X)| < |X| + q$ or determines that there is a~$q$-additive expansion from~$A$ into~$B$ (that is, $|N(X)| \ge |X| + q$ for every~$X \subseteq A$ with~$X \ne \emptyset$).
\end{lemma}
\begin{proof}
	We first show that the function~$f \colon 2^A \to \mathbb{N}$ defined by~$f(X) = |N(X)| - |X|$ for~$X \subseteq A$ is submodular.
	A function~$g\colon 2^A \to \mathbb{N}$ is \emph{submodular} if~$g(X \cup \{ x \}) - g(X) \ge g(X' \cup \{ x \}) - g(X')$ for every~$X \subseteq X' \subseteq A$ and~$a \in A \setminus X'$.
	For the function~$f$ we have
	
	\begin{align*}
		&(f(X \cup \{ x \}) - f(X)) - (f(X' \cup \{ x \}) - f(X')) \\
		&= (|N(X \cup \{ x \})| - |N(X)|) - (|N(X' \cup \{ x \})| - |N(X')|) \\
		&= |N(x) \setminus N(X)| - |N(x) \setminus N(X')|
		= |N(x) \cap N(X')| - |N(x) \cap N(X)| \ge 0.
	\end{align*}
	Thus, $f$ is submodular. 
	
	As the minimum of~$f$ is at~$f(\emptyset) = 0$, minimizing~$f$ is not useful.
	Instead, for each~$a \in A$, we define~$f_a \colon 2^{A} \to \mathbb{N}$ by~$f_a(X) = f(X \cup \{ a \})$ for~$X \subseteq A$.
	As~$f$ is submodular, so is~$f_a$.
	Observe that there exists a nonempty set~$X \subseteq A$ with~$|N(X)| < |X| + q$ if and only if there exists~$a \in A$ with~$\min_{X \subseteq A} f_a(X) < q$.
	Since the minimum of submodular functions can be found in polynomial time~(see e.g.,~\cite{fujishige2005submodular}), the lemma follows.
\end{proof}

A simple, recursive application of \Cref{lemma:additive-expansion} enables us to find a $q$-additive expansion, as shown next.
Note the differences to \cref{lemma:expansion}: the inequality ($|B| \ge q|A|$) becomes strict ($|B| > q|A|$) and as a result one gets a $q$-additive expansion instead of a $q$-expansion.

\begin{lemma}[Additive expansion lemma]\label{lem:add-expansion} 
	Let~$q \in \NN_+$ and~$H$ be a bipartite graph with bipartition~$(A,B)$. If
	\begin{itemize}
		\item $\abs{B} > q \cdot \abs{A}$, and
		\item there are no isolated vertices in~$B$, 
	\end{itemize}
	then there exist non-empty vertex sets~$X \subseteq A$ and~$Y \subseteq B$ such that
	\begin{itemize}
		\item there is a $q$-additive expansion of~$X$ into~$Y$, and
		\item no vertex in~$Y$ has a neighbor outside~$X$, that is,  $N(Y) \subseteq X$.
	\end{itemize}
	Furthermore, the sets~$X$ and~$Y$ are computable in polynomial time. 
\end{lemma}

\begin{proof}
	Given~$H$, we apply \cref{lemma:additive-expansion}.
	If a nonempty set~$X \subseteq A$ such that~$|N(X)| \le |X| + q - 1 \le q |X|$ is found, then we delete~$X$ from~$A$ and~$N(X)$ from~$B$.
	Note that~$X \ne A$ as~$\abs{B} > q \cdot \abs{A}$ holds before deleting~$X$ and~$N(X)$.
	Note also that $\abs{B} > q \cdot \abs{A}$ holds after the deletion of~$X$ and~$N(X)$.
	We repeat this until \cref{lemma:additive-expansion} reports that the remaining, nonempty graph is a $q$-additive expansion.
	Clearly, this process stops after less than~$|A|$ iterations and, by \cref{lemma:additive-expansion}, runs in polynomial time.
\end{proof}

The reduction rule for removing $q$-additive expansions is analogous to \cref{rrule:LinKernelRRuleV0}.
The notable difference is that we require $(\ell+1)$-expansions in \cref{rrule:LinKernelRRuleV0} but $\ell$-additive expansions in \cref{rr:crown}.

\begin{rrule}
	\label{rr:crown}
	If there is an~$\ell$-additive expansion from~$X$ into~$Y$ in~$G[V_1, V_0]$ such that~$N(Y) \subseteq X$, then delete~$X$ and~$Y$ from~$G$ and decrease~$k$ by~$|X|$.
\end{rrule}

\begin{lemma}
  \Cref{rr:crown} is safe.
\end{lemma}
\begin{proof}
  Let~$(G', k', \ell)$ denote the resulting instance.

  ($\Rightarrow$): 
  Let~$S$ be a solution of~$(G, k, \ell)$.
  We claim that~$S' = S \setminus (X \cup Y)$ is a solution of~$(G', k', \ell)$.
  First, observe that~$|E(G' - S')| \le |E(G - S)| \le \ell$ since~$G'-S'$ is a subgraph of~$G-S$.
  We then show that~$|S'| \le k' = k - |X|$.
  Let~$X_S' = X \setminus S$.
  If~$X_S' = \emptyset$, then~$|S'| = |S| - |X| \le k'$.
  So assume that~$X_S' \ne \emptyset$.
  By the assumption, we have~$\abs{N(X_S') \cap Y} \ge |X_S'| + \ell$.
  Since each vertex~$y \in N(X_S') \cap Y$ has at least one neighbor in~$X_S'$, we have~$|S \cap Y| \ge |N(X_S') \cap Y| - \ell \ge |X_S'|$.
  Thus, we have~$|S'| = |S| - |S \cap X| - |S \cap Y| \le |S| - (|X| - |X_S'|) - |X_S'| = |S| - |X| \le k'$, showing that~$S'$ is a solution of~$(G', k', \ell)$.

  ($\Leftarrow$)
  Let~$S'$ be a solution of~$(G', k', \ell)$.
  	Because~$N(y) \subseteq X$ for every~$y \in Y$, we have~$|E(G - (S' \cup X))| = |E(G' - S')| \le \ell$
	Thus,~$S' \cup X$ is a solution of~$(G, k, \ell)$ as~$|S' \cup X| \le k$.
\end{proof}

\begin{theorem}
  \textsc{Partial Vertex Cover} admits a kernel of size~$(\ell' + 1)(k + \ell)$, where~$\ell' = \max \{ \ell, 1 \}$.
\end{theorem}
\begin{proof}
	Given an instance~$(G, k, \ell)$ of \textsc{Partial Vertex Cover}, we first compute the VCLP and exhaustively apply \Cref{rrule:LinKernelRRuleV1V12,rr:crown} exhaustively.

	Denote with~$V_0'$ and~$V_1'$ the vertices remaining in~$V_0$ and~$V_1$.
	Note that~$|V_0'| \le \ell \cdot |V_1'|$ as otherwise~$G[V_0',V_1']$ would contain an~$\ell$-additive expansion (see \cref{lem:add-expansion}) which can be removed by \cref{rr:crown}.
	Thus, the resulting instance has at most

	\begin{align*}
		|V_1'| + |V_0'| + |V_{1/2}|
		&\le (\ell + 1) |V_1'| + |V_{1/2}| \\
		&\le (\ell + 1) |V_1| + |V_{1/2}|
		\le (\ell - 1) |V_1| + 2 (|V_1| + 1/2 \cdot |V_{1/2}|)
	\end{align*}
	vertices.
	Since~$|V_1| \le |V_1| + 1/2 \cdot |V_{1/2}| \le k + \ell$, we have at most~$(\ell + 1) (k + \ell)$ vertices for~$\ell \ge 1$.
\end{proof}

For instance, the number of vertices in the kernel amounts to~$2k + 2$ for~$\ell = 1$.

\section{Conclusion}

We provided a kernel with~$(\ell + 1)(k + \ell)$ vertices for~$\ell \ge 1$ for \textsc{Partial Vertex Cover}.
The first question is whether this can be improved?
For example, does \textsc{Partial Vertex Cover} also admit a kernel with~$2k + \ell^{O(1)}$ vertices?
Besides improving the concrete kernel, finding further applications for the additive expansions seems a promising line of research. 
More precisely, are there more cases where current kernels based on the expansion lemma can be improved with additive expansions?

\section*{Acknowledgement}
Tomohiro Koana is supported by the Deutsche Forschungsgemeinschaft (DFG) project DiPa (NI 369/21).

\bibliography{ref}

\begin{thebibliography}{10}

\bibitem{AFS11}
Omid Amini, Fedor~V. Fomin, and Saket Saurabh.
\newblock Implicit branching and parameterized partial cover problems.
\newblock {\em Journal of Computer and System Sciences}, 77(6):1159--1171,
  2011.
\newblock \href {https://doi.org/10.1016/j.jcss.2010.12.002}
  {\path{doi:10.1016/j.jcss.2010.12.002}}.

\bibitem{CMPS14}
Bugra Caskurlu, Vahan Mkrtchyan, Ojas Parekh, and K.~Subramani.
\newblock On partial vertex cover and budgeted maximum coverage problems in
  bipartite graphs.
\newblock In {\em Theoretical Computer Science - 8th {IFIP} {TC} 1/WG 2.2
  Inter\-nat\-i\-onal Conference, {TCS} 2014}, volume 8705 of {\em Lecture
  Notes in Computer Science}, pages 13--26. Springer, 2014.
\newblock \href {https://doi.org/10.1007/978-3-662-44602-7_2}
  {\path{doi:10.1007/978-3-662-44602-7_2}}.

\bibitem{CKJ01}
Jianer Chen, Iyad~A. Kanj, and Weijia Jia.
\newblock Vertex cover: Further observations and further improvements.
\newblock {\em Journal of Algorithms}, 41(2):280--301, 2001.
\newblock \href {https://doi.org/10.1006/jagm.2001.1186}
  {\path{doi:10.1006/jagm.2001.1186}}.

\bibitem{CFKL+15}
Marek Cygan, Fedor~V. Fomin, Lukasz Kowalik, Daniel Lokshtanov, D{\'{a}}niel
  Marx, Marcin Pilipczuk, Michal Pilipczuk, and Saket Saurabh.
\newblock {\em Parameterized Algorithms}.
\newblock Springer, 2015.
\newblock \href {https://doi.org/10.1007/978-3-319-21275-3}
  {\path{doi:10.1007/978-3-319-21275-3}}.

\bibitem{FRS11}
Fedor~V. Fomin, Daniel Lokshtanov, Venkatesh Raman, and Saket Saurabh.
\newblock Subexponential algorithms for partial cover problems.
\newblock {\em Information Processing Letters}, 111(16):814--818, 2011.
\newblock \href {https://doi.org/10.1016/j.ipl.2011.05.016}
  {\path{doi:10.1016/j.ipl.2011.05.016}}.

\bibitem{fujishige2005submodular}
Satoru Fujishige.
\newblock {\em Submodular functions and optimization}.
\newblock Elsevier, 2005.

\bibitem{GNW07}
Jiong Guo, Rolf Niedermeier, and Sebastian Wernicke.
\newblock Parameterized complexity of vertex cover variants.
\newblock {\em Theory of Computing Systems}, 41(3):501--520, 2007.
\newblock \href {https://doi.org/10.1007/s00224-007-1309-3}
  {\path{doi:10.1007/s00224-007-1309-3}}.

\bibitem{KLR08}
Joachim Kneis, Alexander Langer, and Peter Rossmanith.
\newblock Improved upper bounds for partial vertex cover.
\newblock In {\em Graph-Theoretic Concepts in Computer Science, 34th
  International Workshop, {WG} 2008}, volume 5344 of {\em Lecture Notes in
  Computer Science}, pages 240--251, 2008.
\newblock \href {https://doi.org/10.1007/978-3-540-92248-3_22}
  {\path{doi:10.1007/978-3-540-92248-3_22}}.

\bibitem{KKNS22}
Tomohiro Koana, Christian Komusiewicz, Andr{\'{e}} Nichterlein, and Frank
  Sommer.
\newblock Covering many (or few) edges with k vertices in sparse graphs.
\newblock In {\em Proceddings of the 39th International Symposium on
  Theoretical Aspects of Computer Science, {STACS} 2022}, pages 42:1--42:18,
  2022.
\newblock \href {https://doi.org/10.4230/LIPIcs.STACS.2022.42}
  {\path{doi:10.4230/LIPIcs.STACS.2022.42}}.

\bibitem{Lam11}
Michael Lampis.
\newblock A kernel of order 2 k-c log k for vertex cover.
\newblock {\em Information Processing Letters}, 111(23-24):1089--1091, 2011.
\newblock \href {https://doi.org/10.1016/j.ipl.2011.09.003}
  {\path{doi:10.1016/j.ipl.2011.09.003}}.

\bibitem{MPSW20}
Vahan Mkrtchyan, Garik Petrosyan, K.~Subramani, and Piotr~J. Wojciechowski.
\newblock Parameterized algorithms for partial vertex covers in bipartite
  graphs.
\newblock In Leszek Gasieniec, Ralf Klasing, and Tomasz Radzik, editors, {\em
  Combinatorial Algorithms - 31st International Workshop, {IWOCA} 2020,
  Bordeaux, France, June 8-10, 2020, Proceedings}, volume 12126 of {\em Lecture
  Notes in Computer Science}, pages 395--408. Springer, 2020.
\newblock \href {https://doi.org/10.1007/978-3-030-48966-3\_30}
  {\path{doi:10.1007/978-3-030-48966-3\_30}}.

\bibitem{NT75}
George~L. Nemhauser and Leslie~E. Trotter.
\newblock Vertex packings: Structural properties and algorithms.
\newblock {\em Mathematical Programming}, 8(1):232--248, 1975.
\newblock \href {https://doi.org/10.1007/BF01580444}
  {\path{doi:10.1007/BF01580444}}.

\bibitem{PY22}
Fahad Panolan and Hannane Yaghoubizade.
\newblock Partial vertex cover on graphs of bounded degeneracy.
\newblock {\em CoRR}, abs/2201.03876, 2022.
\newblock URL: \url{https://arxiv.org/abs/2201.03876}, \href
  {http://arxiv.org/abs/2201.03876} {\path{arXiv:2201.03876}}.

\bibitem{RS08}
Venkatesh Raman and Saket Saurabh.
\newblock Short cycles make \emph{W}-hard problems hard: {FPT} algorithms for
  \emph{W}-hard problems in graphs with no short cycles.
\newblock {\em Algorithmica}, 52(2):203--225, 2008.
\newblock \href {https://doi.org/10.1007/s00453-007-9148-9}
  {\path{doi:10.1007/s00453-007-9148-9}}.

\end{thebibliography}

\end{document}